\documentclass{amsart}

\usepackage{amsmath, amsthm, amssymb, mathrsfs, yfonts, braket, xspace}
\usepackage{thm-restate}
\usepackage{xcolor, hyperref}
\usepackage[nameinlink,capitalise]{cleveref}
\hypersetup{
    colorlinks=true,
    linkcolor=red!30!black!70!,
    citecolor=green!30!black!70!,
    filecolor=magenta,      
    urlcolor=red,
    pdftitle={CGT}
  }
\newcommand{\Vars}[1]{{\mathcal X}}
\newcommand{\EN}{{\mathbb N}}
\newcommand{\free}[2]{F_{#2}^{#1}}
\newcommand{\freeL}[2]{\relstr{F}^{#1}(#2)}
\newcommand{\freeR}[3]{\relstr{F}^{#1}_{#3}(#2)}
\newcommand{\minor}[1]{\xrightarrow{#1}}
\newcommand{\Mminor}[3]{\xrightarrow[#2\rightarrow #3]{#1}}
\newcommand{\bla}{\rule{6pt}{.5pt}}

\DeclareMathOperator{\proj}{proj}
\DeclareMathOperator{\id}{id}
\DeclareMathOperator{\Pol}{Pol}
\DeclareMathOperator{\CSP}{CSP}
\DeclareMathOperator{\PCSP}{PCSP}
\DeclareMathOperator{\val}{val}
\DeclareMathOperator{\ex}{ex}
\DeclareMathOperator{\ar}{ar}
\newcommand{\size}[1]{\left | #1 \right |}
\newcommand{\inst}[1]{{\mathcal #1}}

\newcommand{\tuple}[1]{{\mathbf{#1}}}
\newcommand{\relstr}[1]{{\mathbf{#1}}}
\newcommand{\minion}[1]{{\mathcal{#1}}}

\newtheorem{prop}{Proposition}
\newtheorem{theorem}{Theorem}
\newtheorem{definition}{Definition}
\newtheorem{corollary}{Corollary}
\newtheorem{lemma}{Lemma}
\usepackage{todonotes,comment}

\newlength{\problemoffset}
\newcommand{\Psearch}[4]{\begin{list}{}{
\setlength{\leftmargin}{\problemoffset}
\setlength{\rightmargin}{\problemoffset}
\setlength{\parsep}{0pt}
\setlength{\itemsep}{2pt}
\setlength{\topsep}{\itemsep}
\setlength{\partopsep}{\itemsep}
}
\item
{\textsc{#1}}
\item
  {{\sc Instance:} #2}
\item
  {{\sc Promise:} #3}
\item
  {{\sc Goal:} #4}
\end{list}
}

\title[CGT and PCSPs]{Combinatorial Gap Theorem \\
and Reductions between Promise CSPs}

\author{Libor Barto}

\thanks{Libor Barto has received funding from the European Research Council
(ERC) under the European Unions Horizon 2020 research and
innovation programme (grant agreement No 771005).}
\address{Faculty of Mathematics and Physics, 
Charles University, Prague, Czechia}
\email{libor.barto@gmail.com} 

\author{Marcin Kozik}
\address{Theoretical Computer Science Department,
Faculty of Mathematics and Computer Science,
Jagiellonian University, Krakow, Poland} 
\email{marcin.kozik@uj.edu.pl}

\begin{document}

\maketitle

\begin{abstract}
A value of a CSP instance is typically defined as a fraction of constraints that can be simultaneously met. 
We propose an alternative definition of a value of an instance and show that, for purely combinatorial reasons, a value of an unsolvable instance is bounded away from one;
we call this fact a gap theorem.

We show that the gap theorem implies NP-hardness of a gap version of the Layered Label Cover Problem. 
The same result can be derived from the PCP Theorem, but a full, self-contained proof of our reduction is quite short and the result can still provide PCP--free NP--hardness proofs for numerous problems.
The simplicity of our reasoning also suggests that weaker versions of Unique-Games-type conjectures, e.g., the $d$-to-1 conjecture, 
might be accessible and serve as an intermediate step for proving these conjectures in their full strength.   

As the second, main application we provide a sufficient condition under which a fixed template Promise Constraint Satisfaction Problem (PCSP) reduces to another PCSP. The correctness of the reduction hinges on the gap theorem, 
but the reduction itself is very simple.
As a consequence, we obtain that \emph{every} CSP can be canonically reduced to most of the known NP-hard PCSPs, such as the approximate hypergraph coloring problem.
\end{abstract}

\section{Introduction}
Paul would like to find an assignment from $V$ to $A$ -- an element of $A^V$-- that simultaneously satisfies a collection of local constraints. Each constraint demands that the restriction of the assignment onto a subset $W \subseteq V$ of size at most $m$
is in a prescribed subset of $A^W$. We call $V$ and $A$ together with such a collection of local constraints an \emph{instance of $m$-$\CSP$ over $A$}
and denote it by $\Phi$; Paul is looking for a \emph{solution} to $\Phi$.

Paul asks Carole to provide, for some specified $k \geq m$, a collection of partial assignments for $\Phi$: functions $\inst I(U) \in A^U$, where $U$ runs through all $k$-element subsets of $V$~(we write $U\in \binom{V}{k}$), 
such that
\begin{enumerate}
    \item each function $\inst I(U)$ is a partial solution to $\Phi$, 
    i.e., it satisfies every constraint 
    defined on $W \subseteq U$, and
    \item the partial solutions are consistent, i.e., for any $U_1$ and $U_2$, $\proj_{U_1 \cap U_2} \inst I(U_1) = \proj_{U_1 \cap U_2} \inst I(U_2)$~(where $\proj_{U'} \inst I(U)$ denotes the restriction of $\inst I (U)$ to $U'$).
\end{enumerate}
If Carole provides such a collection, must a solution exist? 
Can Paul find a solution given Carole's answer?

The answer to both questions is, trivially, ``Yes''. Indeed, the consistency requirement ensures that all the $\inst I(U)$ are restrictions of a single function $f: V \to A$ and $f$ satisfies all the local constraints since the $\inst I(U)$ are partial solutions and~$k \geq m$.

Let us fix $A$ and natural numbers $m, d, k_0 > k_1$, and make Carol's task easier;
she provides two collections $\inst I_0, \inst I_1$ such that
\begin{enumerate}
    \item for every $U\in\binom{V}{k_i}$, the set $\inst I_i (U)\subseteq A^U$ consists of partial solutions to $\Phi$,
    \item every $\inst I_i(U)$ has no more than $d$-elements, and
    \item if $U_0\supseteq U_1$~(of sizes $k_0, k_1$), then some elements of $\inst I_0(U_0)$ and $\inst I_1(U_1)$ are consistent, i.e.,
    there is $f\in\inst I_0(U_0)$ such that  $\proj_{U_1}f\in \inst I_1(U_1)$.
\end{enumerate}
If Carole provides such collections, must a solution exist? 
Can Paul find a solution given Carole's answer?

Our main theorem provides a positive answer to these questions for each $A$, $m$, $d$ and suitable chosen $k_0$ and $k_1$.  
The property can be concisely stated  in terms of a combinatorial measure of quality of an $m$-CSP instance defined as follows. 
The \emph{$(k_0,k_1)$-value} of an instance $\Phi$, 
denoted $\val_{k_0,k_1}(\Phi)$, is the smallest $d$ for which
Carole can provide consistent collections,
and $\infty$ if no such collections exists.
The positive answer to the first question can now be stated as follows.

  \begin{theorem}[Combinatorial Gap Theorem]\label{thm:cgt}
    For every $A$, $m$ and $d$ there exists $k_0,k_1 \geq m$ such that for every instance $\Phi$ of $m$-$\CSP$ over $A$ either
    \begin{itemize}
      \item $\val_{k_0,k_1}(\Phi) = 1$~(i.e., $\Phi$ is solvable) or
      \item $\val_{k_0,k_1}(\Phi) > d$.
    \end{itemize}
  \end{theorem}

In fact, we prove in \cref{cor:lcgt} (the Layered Combinatorial Gap Theorem) a stronger version that permits more than two collections and only requires a particularly weak form of consistency in the definition of the value of an instance. This fact is in turn a consequence of \cref{thm:main} (the Main Theorem) that does not require the underlying $m$-$\CSP$ instance $\Phi$ in the statement and provides an affirmative answer to the second question -- Paul can compute the solution in polynomial time. 

Our main application is in providing reductions between Promise Constraint Satisfaction Problems. However, let us first discuss the connection to simpler and more standard notions of value.

\subsection{Baby PCP Theorem}

The most straightforward notion of value of an instance $\Phi$ is the following \cite{daBook}: the (standard) \emph{value} of $\Phi$ is
the largest $\varepsilon$ ($0 \leq \varepsilon \leq 1$) such that there exists a function $f: V \to A$ that satisfies $\varepsilon$ fraction of the constraints. 

This is the standard  measure in the area of optimization and approximation algorithms. The following theorem, which follows from the PCP Theorem~\cite{ALMSS98,AS98} and the Parallel Repetition Theorem~\cite{Raz98},
is then a starting point for many NP-hardness results in the area. 

\begin{theorem}[\cite{ALMSS98,AS98,Raz98}] \label{thm:pcp}
For every $1 \geq \varepsilon > 0$ there exists $A$ such that it is NP-hard to distinguish solvable instances of 2-$\CSP$ over $A$  from those whose value is smaller than $\varepsilon$.
\end{theorem}

As an immediate consequence of the Combinatorial Gap Theorem, we obtain a weaker version of \cref{thm:pcp}, which we call the Baby PCP Theorem. Its formulation uses yet another notion of value:
the \emph{combinatorial value} of $\Phi$ is the smallest integer $d$ such that there exists a function $f$ from $V$ to the set of at most $d$-element subsets of $A$ such that, for every local constraint $\varphi \subseteq A^W$, the projection of $f$ onto $W$ intersects $\varphi$. 

\begin{theorem}[Baby PCP Theorem] \label{thm:babyPCP} 
For every integer $d \geq 1$ there exists $A$ such that it is NP-hard to distinguish solvable instances of 2-$\CSP$ over $A$  from those whose combinatorial value is greater than $d$. 
\end{theorem}

Note that \cref{thm:babyPCP} is indeed a consequence of \cref{thm:pcp} by a probabilistic argument that goes as follows. If $f$ witnesses combinatorial value at most $d$ and we define $f'$ by choosing $f'(v)$ from $f(v)$ uniformly at random (independently for each $v$), then the probability that $f'$ satisfies a constraint is at least $1/d^2$ and so is then  the expected fraction of satisfied constraints. Therefore, the trivial reduction (i.e., not changing the input) reduces the problem in \cref{thm:pcp} with $\varepsilon = 1/d^2$ to the problem in \cref{thm:babyPCP}.

On the other, \cref{thm:babyPCP} is still sufficient for some NP-hardness results (such as many known NP-hard PCSPs to be discussed in later sections). 
Our proof of \cref{thm:babyPCP} is based on a very simple reduction from any NP-hard $m$-CSP and a proof of its correctness follows easily from the Main Theorem which is itself not excessively complex.
Most importantly the result suggests that weaker, combinatorial versions of some 
refinements of~\cref{thm:pcp} might be more accessible. We refer to \cref{sec:conlusion} for further discussion.

\subsection{Reductions Between CSPs}

In order to state our main application of the Main Theorem (\cref{thm:main}) we first give some background on the fixed template CSP (in this subsection) and fixed template Promise CSPs (in \cref{subsec:PCSP}). Our contributions are then discussed in \cref{subsec:PCSPnews}. The statements of theorems are informal in that we omit some obvious assumptions and we postpone some definitions to later sections.

The fixed template finite domain CSP is a framework for expressing many computational problems such as various versions of logical satisfiability, graph
coloring, and systems of equations.
A \emph{template} can be specified as a relational structure $\relstr A = (A; R_1, \dots, R_l)$, where $A$ is a finite set called \emph{the domain} and each $R_j$ is a relation of some arity $\ar_j$, i.e., a subset of $A^{\ar j}$. The \emph{CSP over $\relstr{A}$}, denoted $\CSP(\relstr A)$, is (in its search version) the problem of finding an assignment $V \to A$ that satisfies given local constraints as above, with the restriction that each constraint is of the form $\{f \in A^W: (f(w_1),f(w_2), \dots, f(w_{\ar i})) \in R_i\}$, where $1 \leq i \leq l$ and $W = \{w_1, \dots, w_{\ar i}\} \subseteq V$. In the decision version of $\CSP(\relstr A)$ we only want to decide whether such an assignment exists. Our results work for both versions and we do not carefully distinguish between them in the introduction.

Note that for $\relstr{A}$ consisting of all relations on $A$ of arity $m$, the CSP over $\relstr{A}$ is exactly the $m$-CSP over $A$. By choosing appropriate structures with a two-element domain we obtain various versions of satisfiability, such as $k$-SAT, HORN-$k$-SAT, NAE-$k$-SAT, etc. Important class of examples on larger domains is the CSP over $\relstr{K}_n$, the set $[n]=\{1,\dots,n\}$ together with the disequality relation, which is (essentially) the $n$-coloring problem for graphs. More generally, the CSP over $\relstr{NAE}_n^k$, the set $[n]$ together with the $k$-ary not-all-equal relation, is the $n$-coloring problem for $k$-uniform hypergraphs. We refer to surveys in~\cite{KZ17} for more details and examples, as well as many variants of the fixed template CSP framework.

In~\cite{FV98}, Feder and Vardi conjectured that each $\CSP(\relstr A)$ is either solvable in polynomial time or NP-complete. Their conjecture inspired a very active research program in the last 20 years~\cite{BKW17,KZ17}, which culminated in a recent confirmation of the conjecture obtained independently by Bulatov~\cite{Bul17} and Zhuk~\cite{Zhu17,Zhu20}. 

A fundamental theorem, which initiated a rapid development of the subject, crystallized in the series of papers by Jeavons et al., e.g.~\cite{JCG97,Jea98}. It gives a sufficient condition for the existence of a polynomial-time reduction between two CSPs in terms of multivariate functions that preserve relations of the templates, called \emph{polymorphisms} (see \cref{sec:PCSPs}). Denoting $\Pol(\relstr A)$ the set of all polymorphisms of $\relstr A$, the theorem can be stated as follows.

\begin{theorem}[\cite{Jea98}] \label{thm:jeavons}
    If $\Pol(\relstr{A}_1) \subseteq \Pol(\relstr{A}_2)$,  then $\CSP(\relstr{A}_2)$ is reducible to $\CSP(\relstr{A}_1)$. 
\end{theorem}

This theorem was later made more applicable in~\cite{BJK05} and then in~\cite{BOP18} by replacing the inclusion by weaker requirements, thus providing more reductions. A modern formulation is in terms of minion homomorphisms (see \cref{sec:PCSPs}) as follows.

\begin{theorem} \label{thm:CSPred}
  If $\Pol(\relstr{A}_1)$ has a minion homomorphism to $\Pol(\relstr{A}_2)$,  then $\CSP(\relstr{A}_2)$ is reducible to $\CSP(\relstr{A}_1)$. 
\end{theorem}

This theorem has a quite simple proof but it is surprisingly powerful: it follows from Bulatov's and Zhuk's complexity classification~\cite{Bul17,Zhu20} that for any $\CSP(\relstr A)$ either
 \begin{itemize}
     \item $\CSP(\relstr A)$ is solvable in polynomial time \item for \emph{every} $\relstr{B}$, \cref{thm:CSPred} provides a reduction from $\CSP(\relstr{B})$ to $\CSP(\relstr{A})$.
 \end{itemize}
In other words, NP-hard CSPs form the largest equivalence class of the preoder given by the reducibility implied by minion homomorphisms. 
In this sense, \cref{thm:CSPred} provides a single source of hardness and, in fact, the proof of this theorem gives a simple reduction from any CSP to any NP-hard CSP (assuming P $\neq$ NP). The theorem is interesting (but substantially weaker) on the algorithmic side as well; for instance, it gives a reduction of any width 1 CSP~\cite{FV98} to HORN-3-SAT (see \cite{BKO19,BBKO}).

\subsection{Reductions between PCSPs} \label{subsec:PCSP}

The fixed template Promise CSP (PCSP) is a recently introduced generalization of the fixed template CSP, motivated by open problems about (in)approximability of SAT and graph coloring~\cite{AGH17,BG16a,BG18}.
The idea is that each constraint has a strict version and a relaxed version and the problem is (in the search version) to find an assignment satisfying the relaxed constraints given an instance which is satisfiable under the strict constraints (this is the promise). More precisely, the template for PCSP is a pair of similar structures $(\relstr A, \relstr B)$, where $\relstr{A}$
specifies the allowed forms of strict constraints and $\relstr B$ their relaxations (we again refer to \cref{sec:PCSPs} for precise definitions). Note that $\PCSP(\relstr A,\relstr A)$ is the same problem as $\CSP(\relstr A)$.

Important examples of PCSPs include graph coloring and hypergraph coloring problems, such as $\PCSP(\relstr K_n, \relstr K_m)$, $m \geq n$ -- the problem to find an $m$-coloring of a given $n$-colorable graph, and approximate versions of satisfiability problems such as the (2+$\varepsilon$)-SAT problem from~\cite{AGH17}. We refer to~\cite{BG16,BG18,BG18b,BKO19,BBKO} for more examples.

The complexity classification of fixed template PCSPs beyond CSPs is largely unknown; indeed, even the complexity of $\PCSP(\relstr K_n, \relstr K_m)$ is a long-standing open problem~\cite{GJ76} and it is known only for some choices of parameters $n$, $m$ (see~\cite{KOWZ} for a recent account). However, an analogue of \cref{thm:jeavons}~\cite{BG16a} and even \cref{thm:CSPred}~\cite{BKO19,BBKO} is available with a natural generalization of polymorphisms. Denoting $\Pol(\relstr{A},\relstr{B})$ the set of all polymorphisms of a template $(\relstr{A},\relstr{B})$, the latter theorem can be formulated as follows. 

\begin{theorem}[\cite{BKO19,BBKO}] \label{thm:PCSPred}
 If $\Pol(\relstr{A}_1,\relstr{B}_1)$ has a minion homomorphism to $\Pol(\relstr{A}_2,\relstr{B}_2)$,  then $\PCSP(\relstr{A}_2,\relstr{B}_2)$ is reducible to $\PCSP(\relstr{A}_1,\relstr{B}_1)$. 
\end{theorem}

This theorem is still very useful in the more general PCSP setting. For instance, it gives a reduction of any CSP to $\PCSP(\relstr{K}_3,\relstr{K}_4)$ (as essentially shown in~\cite{BG16}, cf.~\cite{BKO19}) and gives a reduction from $\PCSP(\relstr{NAE}_2^3,\relstr{NAE}_{n}^3)$  (i.e., $n$-coloring a 2-colorable 3-uniform hypergraph), for a certain $n$, to $\PCSP(\relstr{K}_3,\relstr{K}_5)$ as shown in~\cite{BKO19,BBKO}. From NP-hardness of the former problem~\cite{DRS05} one obtains NP-hardness of the latter problem. The algorithmic side of this theorem is discussed in Section 7 of~\cite{BBKO}.

However, \cref{thm:PCSPred} is very much insufficient for proving NP-hardness of every NP-hard PCSP, e.g., one provably cannot apply it to reduce an NP-hard CSP (such as 3-SAT) to $\PCSP(\relstr{NAE}_2^3,\relstr{NAE}_n^3)$.

More widely applicable sufficient conditions for NP-hardness in terms of polymorphisms have been developed in~\cite{BBKO,BWZ20,GS20},  or follow from the results in these papers. They are all based on \cref{thm:pcp} and its refinements. These conditions cover almost all known NP-complete PCSPs, a notable exception being~\cite{Hua13}.

On the other hand, these sufficient conditions are not quite satisfactory for two reasons. First, they are not based on a general reduction theorem such as \cref{thm:PCSPred}, which limits their applicability and appeal. Second, they use complex NP-hardness results (\cref{thm:pcp} and refinements), which, e.g., makes it difficult to reduce a standard NP-complete problem, like 3-SAT, to many NP-hard PCSPs. For instance, if we want to reduce 3-SAT to $\PCSP(\relstr{K}_3,\relstr{K}_5)$ using available theory,  we first need to perform a sequence of reductions used in a proof of the PCP theorem, then another reduction for the Parallel Repetition Theorem, ending up in the situation of \cref{thm:pcp}, then further reductions for an improved version of \cref{thm:pcp} from~\cite{Kho02}, followed  by reductions to approximate hypergraph coloring from~\cite{DRS05}, finally finishing with  reductions provided by \cref{thm:PCSPred} to $\PCSP(\relstr{K}_3,\relstr{K}_5)$~\cite{BKO19}. 
Such a long chain of reductions obscures the reasons why the problem is hard.

\subsection{New reductions between PCSPs} \label{subsec:PCSPnews}

We define a concept of a minion $(d,r)$-homomorphism (\cref{def:drhomo}) that weakens minion homomorphisms in the following sense: for $d=1, r=1$ the concepts coincide, and increasing $d$ or $r$ makes the concept weaker. We then apply the Main Theorem (\cref{thm:main}) to show that a generalization of \cref{thm:PCSPred} remains true with this weaker concept, thus giving us more reductions between PCSPs.

\begin{theorem} \label{thm:PCSPbetterRed}
If $\Pol(\relstr{A}_1,\relstr{B}_1)$ has a minion $(d,r)$-homomorphism to $\Pol(\relstr{A}_2,\relstr{B}_2)$,  then $\PCSP(\relstr{A}_2,\relstr{B}_2)$ is reducible to $\PCSP(\relstr{A}_1,\relstr{B}_1)$. 
\end{theorem}

This theorem partially resolves the shortcomings of the state-of-the-art discussed above. In particular, the theorem gives a reduction of any NP-hard CSP to many known NP-hard PCSPs, including
\begin{itemize}
    \item all NP-hard Boolean symmetric PCSPs, which were classified in~\cite{FKOS19} (e.g., the $(2+\varepsilon)$-SAT from~\cite{AGH17} and, more generally, NP-hard symmetric folded PCSPs classified in~\cite{BG18}),
    \item those NP-hard approximate coloring problems, i.e., PCSPs of the from $\PCSP(\relstr{K}_n,\relstr{K}_m)$, identified in \cite{BKO19} (e.g., $\PCSP(\relstr{K}_3,\relstr{K}_5)$),
    \item all approximate 3-uniform hypergraph coloring problems, i.e., PCSPs over $(\relstr{NAE}_n^3,\relstr{NAE}_m^3)$ \cite{DRS05} (for this we need an improvement of the proof by Wrochna~\cite{wrochna}),
    \item those NP-hard PCSPs identified in~\cite{BWZ20} (promise SAT on non-Boolean domains), in~\cite{KO19} (graph 3-coloring  with strong promises), and in~\cite{BBB21} (variants of 3-uniform hypergraph coloring).
\end{itemize}

The examples of reductions that are not (known to be) covered include NP-hardness proofs 
in~\cite{GS20,ABP20,Hua13} 
and reductions that were used in~\cite{WZ20} to improve~\cite{Hua13}.  These examples suggest directions for improving the Main Theorem and thus \cref{thm:PCSPbetterRed}; we discuss these directions in the Conclusion.

The reduction in \cref{thm:PCSPbetterRed} is very simple and the proof of correctness essentially amounts to applying the Main Theorem, whose proof is itself quite short. We now explain the reduction in some detail.

We fist observe that every $\PCSP(\relstr{A},\relstr{B})$ is equivalent to a certain  $\PCSP(\relstr{A}',\relstr{B}')$, where $\relstr{A}'$ consists of \emph{all} relations up to any fixed sufficiently large arity on any sufficiently large domain.  This is a very simple consequence of ~\cite{BKO19} but still a remarkable observation: we can use any instance as an input to $\PCSP(\relstr{A}',\relstr{B}')$ and thus effectively to $\PCSP(\relstr{A},\relstr{B})$.  In fact, the trivial reduction from $\CSP(\relstr{A}_2,\relstr{B}_2)$ to $\CSP(\relstr{A}_1',\relstr{B}_1')$ is correct in the situation of \cref{thm:PCSPred} (and, again, this was essentially proved in~\cite{BKO19}). Our reduction is just the next most obvious one -- in essence, we introduce a variable for every bounded arity subset of the original variables and include the obvious constraints coming from the requirement that values of variables form partial solutions. 

In summary, the reduction from $\PCSP(\relstr{A_2},\relstr{B_2})$  to $\PCSP(\relstr{A}_1,\relstr{B}_1)$ is a composition of a reduction from  $\PCSP(\relstr{A}_2,\relstr{B}_2)$ to $\PCSP(\relstr{A}_1',\relstr{B}_1')$ (which is the ``repetition'' reduction describe above) and a reduction from the latter PCSP to $\PCSP(\relstr{A}_1,\relstr{B}_1)$ (which is the ``polymorphism'' or ``long code '' reduction). We remark that all the reductions in the PCSP/PCP area, which we are aware of, are variations of these two types of reductions. Is it a coincidence?

Finally, it seems unlikely that \cref{thm:PCSPred} is a single source of hardness for all PCSPs in the same sense as \cref{thm:CSPred} is for CSPs. However, we hope that our result will serve as a useful step toward the goal of obtaining such a theorem, which would give a uniform reduction that completely replaces (and explain) a bit ad hoc intermediate problems and reductions that are still necessary for some PCSPs. Ideally, and this seems much more challenging even for CSPs, the theorem would also fully capture the tractability part. Another exciting direction is toward the more general Promise Valued CSP framework  (see~\cite{AH13,VZ20}), which includes problems such as those in \cref{thm:pcp}. Remarkably,  an analogue of \cref{thm:PCSPred} is already available by an unpublished work of Kazda~\cite{kazda}.

\section{Main Theorem}
    This section is devoted to introducing notation and stating the
    main result of the paper in the full strength.
    First we formalize the notion of the information that is provided by Carole.
    For a set of variables $V$ and domain $A$, 
    a \emph{partial assignment system}~(\emph{PAS})
    of arity $k$~(\emph{$k$-PAS}) is a map from the set of all $k$-element subsets of $V$ such that, for each $ U\in\binom{V}{k}$, we have
    $\emptyset \neq \inst I(U)\subseteq A^U$.
    An assignment $f\in A^V$ is an \emph{$m$-solution} of a $k$-PAS $\inst I$, if every $U\in \binom{V}{m}$ can be extended to $W\in\binom{V}{k}$ while satisfying
    $\proj_U f\in\proj_U \inst I(W)$. 
    The \emph{value} of a PAS $\inst I$
    is the maximal size of $\inst I(U)$.
    
    Let $(\inst I_0,\dotsc,\inst I_r)$ be partial assignment systems over common $V$ and $A$.
    We call such a sequence \emph{consistent} if
    \begin{itemize}
        \item their arities $k_0,\dotsc,k_r$ form a non-increasing sequence, and 
        \item for every $U_0\supseteq \dotsb\supseteq U_r$~(of sizes $k_0,\dotsc, k_r$)
        there exists $i<j$ such that $\inst I_j(U_j)\cap \proj_{U_j}\inst I_i(U_i)\neq\emptyset$.
    \end{itemize}
    The \emph{value} of such a sequence is the maximal among values of $\inst I_i$'s.

   \begin{restatable}[Main Theorem]{theorem}{main} 
    \label{thm:main}
        For any $A$ and any numbers $m,r,d \in \EN$ there exists a sequence $k_0,\dots,k_r$ such that if $(\inst I_0,\dotsc,\inst I_r)$ is a consistent sequence of arities $k_0,\dotsc,k_r$ and value $\leq d$, then some $\inst I_i$ has an $m$-solution. 
        Additionally, for fixed $m, r, d, A$, an $m$-solution can be computed in polynomial time.\footnote{Our procedure is very much non-polynomial with respect to parameters $|A|,m, r$ or $d$.}
    \end{restatable}
    
  A proof of this theorem is provided in \cref{sec:proofPCP}.    
   
   Let us revisit the Paul/Carole interaction.
   Given an instance $\Phi$, Carole is providing two consistent  PASes
   containing local solutions to $\Phi$.
   Clearly Paul can make Carole's task easier, by asking for longer sequences.
   The Combinatorial Gap Theorem (\cref{thm:cgt}) can be generalized to accommodate such extensions.
   
   For $\Phi$, an $m$-$\CSP$ instance over $V$ and $A$, and $k_0\geq k_1\geq \dotsb\geq k_r \geq m$
   we can put the $\val_{k_0,\dotsc,k_r}(\Phi)$
   to be the \emph{smallest} value of a consistent sequence $(\inst I_0,\dotsc,\inst I_r)$~(over $V$ and $A$) of arities $k_0,\dotsc,k_r$ such that every element of  $\inst I_i(U)$ is a  partial solution to $\Phi$.
   The following strengthening of the Combinatorial Gap Theorem
   follows immediately from \cref{thm:main}.

   \begin{corollary}[Layered Combinatorial Gap Theorem]\label{cor:lcgt}
   For every $A$ and numbers $m, r, d \in \EN$ there exists $k_0\geq\dotsb\geq k_r \geq m$ such that for every instance $\Phi$ of $m$-$\CSP$ over $A$ either
    \begin{itemize}
      \item $\val_{k_0,\dotsc,k_r}(\Phi) = 1$~(i.e. $\Phi$ is solvable) or
      \item $\val_{k_0,\dotsc,k_r}(\Phi) > d$.
\end{itemize}
\end{corollary}
   \begin{proof}
     Let $k_0,\dotsc,k_r$ be the numbers provided by \cref{thm:main} for $A$ and $m,r,d$.
     Let $\Phi$ be an instance such that
     $\val_{k_0,\dotsc,k_r}(\Phi)\leq d$ and let a sequence $\inst I_0,\dotsc,\inst I_r$ provides this value.
     By \cref{thm:main}
     there exists an $m$-solution to a $k_i$-PAS $\inst I_i$ for some $i$.
     Since $\inst I_i$ consists of partial solutions to $\Phi$ and $\Phi$ is an $m$-$\CSP$ instance, the $m$-solution to $\inst I_i$ is in fact a solution to $\Phi$.
     Thus $\Phi$ is solvable and $\val_{k_0,\dotsc,k_r}(\Phi)=1$.
   \end{proof}

\section{Baby Layered PCP Theorem}

In this section we formulate an improvement of \cref{thm:pcp} that was essentially proved in~\cite{DGKR05} and adapted to this form in~\cite{BWZ20}. Then we show that a weaker, combinatorial version of this theorem (which is a stronger version of the Baby PCP Theorem from the introduction) is a straightforward consequence of the Layered Combinatorial Gap Theorem. 

For convenience we define Layered Label Cover in  a somewhat less standard way in that we allow different domains of variables. The difference is inessential.  

An \emph{$r$-Layered Label Cover} instance consists of 
\begin{itemize}
    \item a set $X$ of \emph{variables}, which is a disjoint union of sets $X_0$, \dots, $X_r$ (called \emph{layers}),
    \item a set $A_x$ for each $x \in X$, called the \emph{domain} of $x$,
    \item a set of \emph{constraints} of the form $((x,y),\psi)$, where $x \in X_i$ and $y \in X_j$ for some $i<j$, and $\psi$ is a map $A_x \to A_y$. We refer to such a constraint as a constraint from $x$ to $y$ and require that there is at most one constraint from $x$ to $y$ for any pair of variables $x$, $y$.  
\end{itemize}
An \emph{assignment} for such an instance is a mapping $f$ with domain $X$ such that $f(x) \in A_x$ for every $x \in X$. It \emph{satisfies} a constraint $((x,y),\psi)$ if $\psi(f(x))=f(y)$. 
A \emph{chain} is a sequence of variables $(x_0, \dots, x_r)$, $x_i \in X_i$, such that there is a constraint from $x_i$ to $x_j$ for each $i<j$. It is \emph{weakly satisfied} by an assignment $f$ if $f$ satisfies at least one of the constraints from $x_i$ to $x_j$, $i < j$.
Finally, the \emph{layered value} of an instance is the largest $\varepsilon$ such that there exists an assignment that weakly satisfies at least $\varepsilon$ fraction of all chains. 

\begin{theorem}[\cite{BWZ20}, Layered PCP Theorem] \label{thm:LPCP}
For every $r \in \EN$ and  $\varepsilon > 0$ there exists  $N \in \EN$ such that,
in the set of instances of $r$-Layered Label Cover with domain sizes at most $N$, 
it is NP-hard to distinguish solvable ones
from those whose layered value is smaller than~$\varepsilon$.
\end{theorem}

A combinatorial adaption of layered value goes as follows.
For each variable we allow $d$-choices of values; formally
a \emph{$d$-assignment} for an $r$-Layered Label Cover instance is a mapping with domain $X$ such that each $f(x)$ is a subset of $A_x$ of size at most $d$.
Then we generalize the notion of weak satisfiability in the most natural way: 
a chain $(x_0,\dotsc,x_r)$ is \emph{weakly satisfied} by a $d$-assignment  
$f$ if for some $i<j$ the constraint $((x_i,x_j),\psi)$ is such that $\psi(f(x_i))\cap f(x_j)\neq\emptyset$.
Finally, the \emph{combinatorial layered value} of an instance is the smallest $d$ such that there exists a $d$-assignment that weakly satisfies all the chains.

\begin{theorem}[Baby Layered PCP Theorem] \label{thm:babyLPCP}
For every $r \in \EN$ and  $\varepsilon > 0$ there exists  $N \in \EN$ such that,
in the set of instances of $r$-Layered Label Cover with domain sizes at most $N$,
it is NP-hard to distinguish solvable ones
from those whose combinatorial layered value is greater than $d$.
\end{theorem}

\begin{proof}
   We fix $r,d$ and reduce from the $m$-CSP over $A$ for any fixed $m$ and $A$, which is enough since, e.g., $3$-CSP over $\{0,1\}$ is NP-hard. 
   
   Let $k_0, \dots, k_r$ be the numbers provided by~\cref{cor:lcgt} and 
   let $\Phi$ be an instance of $m$-CSP over $A$ with a set of variables $V$. We define an instance $\Psi$ of $r$-Layered Label Cover as follows.
   The $i$-th layer variable set is defined as $X_i = \binom{V}{k_i}$ and the domain $A_U$ of $U \in X_i$ is defined as the set of all partial solutions $U \to A$ of $\Phi$ (i.e., a variable $U \in X_i$ is a $k_i$-element set of the original variables and its domain is a subset of $A^U$). For each $U\supseteq W$ we include a constraint $((U,W),\psi)$ with $\psi: A_U \to A_W$ defined by $\psi(g) = \proj_W(g)$. Notice that the definition of $\psi$ makes sense since a restriction of a partial solution to $\Phi$ is a partial solution to $\Phi$. This finishes the construction.
   
   Soundness of this reduction is immediate: if $h: V \to A$ is a solution to $\Phi$, then $f(U) = \proj_U h$ defines a solution to $\Psi$. To prove completeness, assume that $f$ is a $d$-assignment for $\Psi$ that weakly satisfies all the chains. For $0 \leq i \leq r$ and $U \in X_i$ define $\inst I_i(U) = f(U)$ and note that, by construction of $\Psi$, $\inst I_i$ is a $k_i$-PAS for $\Phi$ and, since $f$ is a $d$-assignment that weakly satisfies chains, the sequence $(\inst{I}_0, \dots, \inst{I}_r)$ is consistent. Therefore $\val_{k_0, \dots, k_r}(\Phi) \leq d$ and the Layered Combinatorial Gap Theorem (\cref{cor:lcgt}) finishes the proof by showing that  $\Phi$ is solvable.
   \end{proof}

\section{Promise Constraint Satisfaction Problems}
\label{sec:PCSPs}

In this section we formally define fixed template PCSPs, their polymorphism minions, and minion homomorphisms -- the concepts that are necessary to fully understand the statement of \cref{thm:PCSPred}.

We start by defining homomorphisms between relational structures. We will only work with finite relational structures of finite signature, therefore we can use the formalism from the introduction, that is, a \emph{relational structure} $\relstr{A}$ is a tuple $\relstr A =(A;R_1,\dotsc,R_l)$, where $A$ is a finite \emph{domain} and $R_i \subseteq A^{\ar i}$ is a nonempty relation of \emph{arity} $\ar i$. Two structures are \emph{similar} if they have the same number of relations and corresponding relations have the same arity. For two similar structures $\relstr{A} = (A; R_1, \dotsc R_l)$ and $\relstr{B} = (B; S_1, \dotsc, S_l)$, a \emph{homomorphism} for $\relstr{A}$ to $\relstr{B}$ is a map $h: A \to B$ that preserves the relations, i.e., for any $i$ and any tuple $\tuple{a} \in R_i$, the tuple $h(\tuple{a})$, obtained by component-wise application of $h$, is in $S_i$.

A \emph{CSP template} is a relational structure. The CSP over $\relstr{A}$ is defined by allowing only the CSP instances over $A$ such that each constraint is, in essence, one of the $R_j$. Formally, each constraint $\varphi \subseteq A^W$ is equal to
 $\{f \in A^W: (f(w_1),f(w_2), \dots, f(w_{\ar j})) \in R_j\}$, where $1 \leq j \leq l$ and $W = \{w_1, \dots, w_{\ar j}\}$. 
For notation's sake, we identify\footnote{Note that neither the sequence $(w_1,\dotsc,w_{\ar_j})$ nor the relation $R_j$ needs to be uniquely determined by $\varphi_i$. On the other hand $\varphi_i$ is determined by $(w_1,\dotsc,w_{\ar_j})$ and $R_j$.}
$\varphi$ with the pair $((w_1,\dotsc,w_{\ar_j}),R_j)$.

\subsection{Promise CSPs}
A \emph{PCSP template} is a pair $(\relstr{A},\relstr{B})$ of similar relational structures such that there exists a homomorphism from $\relstr{A}$ to $\relstr{B}$. 
Denoting $\relstr{A}=(A; R_1, \ldots, R_l)$ and $\relstr{B}=(B; S_1, \ldots, S_l)$, 
the PCSP over such a template is defined as follows.

\Psearch{Promise CSP: $\PCSP(\relstr A,\relstr B)$}
{a set of formal constraints of the form $((w_1,\dotsc,w_{\ar_j}),R_j/S_j)$}
{instance with constraints $((w_1,\dotsc,w_{\ar_j}),R_j)$ is solvable}
{find a solution to the instance with constraints $((w_1,\dotsc,w_{\ar_j}),S_j)$}

Given an instance $\Phi$ of $\PCSP(\relstr A,\relstr B)$, the instance of $\CSP(\relstr A)$ appearing in the promise is denoted $\Phi^{\relstr{A}}$ and referred to as the \emph{strict version of $\Phi$}. Similarly, the instance of $\CSP(\relstr B)$ in the goal is the \emph{relaxed version of $\Phi$}, denoted $\Phi^{\relstr{B}}$.

The existence of a homomorphism $h: \relstr{A} \to \relstr{B}$ is sufficient (and necessary) to guarantee that $\PCSP(\relstr{A},\relstr{B})$ makes sense: if the promise is fulfilled, i.e., $\Phi^{\relstr{A}}$ has a solution $f: V \to A$, then the goal can be reached, i.e., $\Phi^{\relstr{B}}$ has a solution, namely $hf$.

We have defined the fixed template PCSP in its \emph{search version}.
In the \emph{decision version of $\PCSP(\relstr{A},\relstr{B})$}, the task is to distinguish instances $\Phi$ solvable in $\relstr{A}$ (i.e., $\Phi^{\relstr{A}}$ is solvable) from those that are not even solvable in $\relstr{B}$. We present our reductions for the official, search version of the problem, which clearly gives us reductions for the decision version as well.

\subsection{Polymorphism minions}

Let $t:A^n\rightarrow B$ and $\pi:[n]\rightarrow [m]$. We say that $s:A^m\rightarrow B$ is a \emph{minor}~(or $\pi$-minor, if $\pi$ matters) of $t$ and write $t\minor{\pi} s$ if $s(a_1,\dotsc,a_m)=t(a_{\pi(1)},\dotsc,a_{\pi(n)})$ for any $(a_1, \dotsc, a_m) \in A^m$.

\begin{definition}[minion]
A \emph{minion} $\minion M$ on a pair of sets $(A,B)$ is a subset of $\bigcup_{i \geq 1} B^{A^i}$ such that 
\begin{itemize}
    \item $\minion M\neq\emptyset$, and
    \item if $t\in\minion M$ and $t\minor{\pi} s$ for some suitable $\pi$, then $s\in\minion M$.
\end{itemize}
\end{definition}

An $n$-ary \emph{polymorphism} of a PCSP template $(\relstr A, \relstr B)$ is a map $t: A^n \to B$ such that for any relation $R_i$ of $\relstr{A}$ and any $[\ar i] \times [n]$ matrix whose columns are in $R_i$, the tuple obtained by applying $t$ to the rows is in the corresponding relation $S_i$ of $\relstr{B}$. 
The set of all polymorphims of a template $(\relstr A,\relstr B)$ is denoted by $\Pol(\relstr A,\relstr B)$.
It is easy to observe~(cf. \cite{BBKO}) that $\Pol(\relstr A,\relstr B)$ is a minion.

The final concept required for \cref{thm:PCSPred} is minion homomorphism.

\begin{definition}[minion homomorphism; Definition 2.21~\cite{BBKO}]
    Let $\minion M, \minion N$ be two minions~(not necesarilly on the same pairs of sets). 
    A mapping $\xi:\minion M\rightarrow\minion N$ is called a \emph{minion homomorphism} if
    \begin{enumerate}
        \item it preserves arities, i.e., arity of $\xi(t)$ is equal to arity of $t$ for all $t\in \minion M$, and
        \item it preserves taking minors i.e. if $t\minor{\pi} s$ then
        $\xi(t)\minor{\pi}\xi(s)$.
    \end{enumerate}
\end{definition}

We are ready to formally state \cref{thm:PCSPred}.

\begin{theorem}[Theorem 3.1~\cite{BBKO}]\label{thm:oldred} 
    Let $(\relstr A_1,\relstr B_1)$ and $(\relstr A_2,\relstr B_2)$ be two  PCSP templates, and let $\minion M_i=
    \Pol(\relstr A_i,\relstr B_i)$ for $i=1,2$.
    If there exists a minion homomorphism $\xi:\minion M_1\rightarrow\minion M_2$
    then $\PCSP(\relstr A_2,\relstr B_2)$ is log-space reducible to $\PCSP(\relstr A_1,\relstr B_1)$.
\end{theorem}

\section{New reduction for PCSPs}
    In this section we formally state our main application of the Layered Combinatorial Gap Theorem, \cref{thm:PCSPbetterRed} and mention some consequences.
    
    A \emph{chain of minors}, which is a useful notion we borrow from~\cite{BWZ20}, is a sequence of minors 
    $t_0\minor{\pi_{0,1}}t_1\minor{\pi_{1,2}}t_2\minor{\pi_{2,3}}\dotsb\minor{\pi_{r-1,r}}t_r$. For such a sequence and $i<j$ we denote by $\pi_{i,j}$ the composition $\pi_{j-1,j}\circ\dotsb\circ\pi_{i,i+1}$; observe that $t_i\minor{\pi_{i,j}}t_j$. 
    The new concept of $(d,r)$-minion homomorphism is defined by requiring a weak form of preservation of chains as follows.

    \begin{definition}[$(d,r)$-minon homomorphism] \label{def:drhomo}
    Let $\minion M, \minion N$ be two minions and $d,r \in \EN$.
A mapping $\xi$ from $\minion M$ to the set of all at most $d$-element subsets of $\minion N$ is called a \emph{minion $(d,r)$-homomorphism} if
    \begin{enumerate}
        \item it preserves arities, i.e., every $g \in \xi(t)$ has the same arity as $t$; and
        \item for any chain of minors 
        $t_0\minor{\pi_{0,1}}t_1\minor{\pi_{1,2}}t_2\minor{\pi_{2,3}}\dotsb\minor{\pi_{r-1,r}}t_r$ there
        \begin{equation*}
            \text{exists } i< j \text{ and } g\in\xi(t_i)\ h\in\xi(t_j) \text{ satisfying } g\minor{\pi_{i,j}} h. 
        \end{equation*}

    \end{enumerate}
    \end{definition}

     Notice than minion $(1,1)$-homomorphism is essentially the same as minion homomorphism and that the concept of $(d,r)$-homomorphism gets weaker as $d$ or $r$ increase. We also remark that $(d,r)$-homomorphisms can be composed with  $(1,1)$-homomorphisms from either side\footnote{It is, however, unclear to us whether the composition of two $(d,r)$-homomorphisms (for some $d,r$) is a $(d',r')$-homomorphism.}.

    The following formal statement of \cref{thm:PCSPbetterRed} is obtained by replacing minion homomorphisms in \cref{thm:oldred} by this weaker concept. 
    A proof is in \cref{sec:proofRed}.

\begin{restatable}{theorem}{drHomoReduction} \label{thm:drhomo}
Let $(\relstr A_1,\relstr B_1)$ and $(\relstr A_2,\relstr B_2)$ be two  PCSP templates, and let $\minion M_i=
    \Pol(\relstr A_i,\relstr B_i)$ for $i=1,2$.
    If there is a minion $(d,r)$-homomorphism $\xi:\minion M_1\rightarrow\minion M_2$~(for some $d$ and $r$)
    then $\PCSP(\relstr A_2,\relstr B_2)$ is P-time reducible\footnote{We believe that the reduction can be done in log space, but do not include the details here.}
    to $\PCSP(\relstr A_1,\relstr B_1)$.
\end{restatable}

The condition for NP-hardness of $\PCSP(\relstr{A},\relstr{B})$ stated as Corollary 4.2. in~\cite{BWZ20} is equivalent to requiring that $\minion{M} = \Pol(\relstr{A},\relstr{B})$ has a $(d,r)$-homomorphism to the trivial minion $\minion{T}$ consisting of all the dictators on some (any) set $C$ of size at least 2 (a \emph{dictator} is the function $(c_1, \dots, c_n) \mapsto c_i$ for some $i \leq n$). \cref{thm:drhomo} additionally provides a reduction from $\PCSP(\relstr{A}',\relstr{B}')$ to $\PCSP(\relstr{A},\relstr{B})$ for any template $(\relstr{A}',\relstr{B}')$ whose polymorphism minion has a homomorphism to $\minion{T}$, such as any NP-hard CSP.

A special situation when a $(d,r)$-homomorphism with $r=1$ from $\minion{M}$ to $\minion{T}$ exists is when $\minion{M}$ does not contain a constant map and all members of $\minion{M}$ depend on at most $d$ coordinates (the homomorphism assigns to $t$ the dictators corresponding to the coordinates that $t$ depends on). This special situation is already quite useful for NP-hardness results~(see \cite{BBKO,KOWZ}). 

A new general consequence we can derive from \cref{thm:drhomo} is that, roughly, the complexity of a PCSP does not depend on low arity polymorphisms. More precisely, if two polymorphism minions differ only in functions that depend on bounded number of coordinates, then the corresponding PCSPs have the same complexity.

\section{Conclusion} \label{sec:conlusion}

We have shown that solutions to CSP instances can be reconstructed from weakly consistent small systems of partial solutions.

The first application was in showing a combinatorial version of (Layered) PCP Theorem, the Baby (Layered) PCP Theorem. One open question is whether there is a combinatorial analogue of the Parallel Repetition Theorem~\cite{Raz98}, in particular, whether the tame dependence of domain size on the value in Raz's result can be achieved in the combinatorial version (note that the dependence in the presented version is rather wild). Another direction is exploring combinatorial versions of known improvements of the PCP Theorem, in particular the Smooth Label Cover of Khot~\cite{Kho02} (cf. \cite{GS20}). 
Finally, the most interesting direction seems to be in exploring combinatorial versions of conjectural improvements of the PCP Theorem, e.g., the $d$-to-1 Conjecture~\cite{Kho02UG}. One of the combinatorial versions of this conjecture is the problem in \cref{thm:babyLPCP} restricted to $r=1$ and instances where every constraint $\psi$ is given by a $d$-to-1 map.

The second, main application of the main result was in providing a general  condition for the existence of a polynomial time reduction between two  PCSPs in terms of polymorphisms -- symmetries of the template. This, and similar such results should not be regarded as heavy hammers that are giving us reductions for free. They rather serve as tools that enable one to disregard the inessential layers  and concentrate on the core of the problem, which can then be  attacked using various methods 
(such as algebraic~\cite{Bul17,Zhu20}, topological~\cite{DRS05,BBKO,KOWZ}, or analytic~\cite{GS20}). As such tools, they \emph{are} indeed useful. 

Moreover, such a general condition seems \emph{necessary} for a prospective dichotomy result for PCSPs, since ``the non-existence of [some specific kind of a] homomorphism to the trivial minion $\minion{T}$'' can potentially be translated to a positive property that can be exploited by an algorithm (as was done in the CSP context~\cite{Bul17,Zhu20}), whereas ``the non-existence of series of tricks proving NP-hardness'' lacks this potential. 

We do not believe that $(d,r)$-homomorphism is already the right, sufficiently week concept. A concrete direction for an improvement is, besides the directions mentioned above,  to incorporate the reduction in~\cite{KOWZ} via an adjunction, which was used to significantly enlarge the NP-hardness region for the approximate graph coloring problem. It is interesting that the reduction in the proof of~\cref{thm:drhomo} works, but it does not seem to be explained by $(d,r)$-homomorphisms.

Another appealing direction for generalizing the reduction theorem is to ``let the Baby PCP Theorem grow up'', i.e. to consider weighted relations (cost functions), where tuples can have weights instead of just being present or absent. 
It may be challenging to obtain such a generalization (as a satisfactory analogue would cover, e.g., the PCP Theorem) 
but there are some indications that such a result is not out of reach: an analogue of~\cref{thm:oldred} is available~\cite{kazda} and Dinur's proof of the PCP Theorem~\cite{Din07} uses essentially the same two reductions as~\cref{thm:drhomo} -- they are substantially fine-tuned and repeated more times, but the essence \emph{is} the same.

\section{Appendix: Proof of \texorpdfstring{\cref{thm:main}}{the Main Theorem}} \label{sec:proofPCP}
For reader's convenience we recall the basic definitions and notations that appear in the proof.
The set of variables is denoted by $V$, while the domain is $A$.
By $\binom{V}{k}$ we denote the set of all $k$-element subsets of $V$,
and for a function $f:V\rightarrow A$ and $U\subseteq V$, the restriction of $f$ to $U$
will be denoted by $\proj_U f$; 
the same notation applies to sets of functions.

A \emph{partial assignment system}~(\emph{PAS})
of arity $k$~(\emph{$k$-PAS}) is a map such that for each $U\in\binom{V}{k}$ we have
$\emptyset \neq\inst I(U)\subseteq A^U$.
An $f\in A^V$ is an \emph{$m$-solution} of a $k$-PAS $\inst I$, 
if every $U\in \binom{V}{m}$ can be extended to $W\in\binom{V}{k}$ satisfying
$\proj_U f\in\proj_U \inst I(W)$. 
The \emph{value} of a PAS $\inst I$
is the maximal size of $\inst I(U)$.

Let $(\inst I_0,\dotsc,\inst I_r)$ be a sequence of partial assignment systems over common $V$ and $A$.
We call such a sequence \emph{consistent} if
\begin{itemize}
    \item their arities $k_0,\dotsc,k_r$ form a non-increasing sequence, and  
    \item for every $U_0\supseteq \dotsb \supseteq U_r$~(of sizes $k_0,\dotsc, k_r$)
    there exists $i<j$ such that $\inst I_j(U_j)\cap \proj_{U_j}\inst I_i(U_i)\neq\emptyset$.
\end{itemize}
The \emph{value} of such a sequence is the largest among values of $\inst I_i$.
Finally, we recall the theorem we are proving:

\main*

\subsection{Working with a single PAS}
  Let $\inst I$ be a $k$-PAS over $V$ and $A$.
  Let $X\subseteq V$ satisfy $|X|\leq k$ and $f\in A^X$;
  for $l\leq k$ we consider two $l$-properties a pair $(X, f)$ can have:
  \begin{align*}
    (P) \quad &\forall W\in \tbinom{V}{l}\ \exists U\in \tbinom{V}{k}\  X\cup W\subseteq U \text{ and } \proj_X g = f \text{ for some } g\in\inst I(U)\\
    (Q) \quad &\forall W\in \tbinom{V}{l}\ \exists U\in \tbinom{V}{k}\  X\cup W\subseteq U \text{ and } \proj_X g \neq f \text{ for all } g\in\inst I(U)
  \end{align*}
  For convenience, we state their negations as well:
  \begin{align*}
    (\neg P) \quad &\exists W\in \tbinom{V}{l}\ \forall U\in \tbinom{V}{k}\   X\cup W\subseteq U \text{ implies } \proj_X g \neq f \text{ for all } g\in\inst I(U)\\
    (\neg Q) \quad &\exists W\in \tbinom{V}{l}\ \forall U\in \tbinom{V}{k}\   X\cup W\subseteq U \text{ implies } \proj_X g = f \text{ for some } g\in\inst I(U)
  \end{align*}
  If $X=\set{v}$ and $f(v)=a$ we say that the property holds for $(v,a)$ instead of $(X,f)$;
  the sets $W$ in the definitions of $\neg Q, \neg P$ are called \emph{witnesses}.

  The first proposition states that if the parameters are suitable chosen,  property $P$ will appear.
  \begin{prop}\label{prop:allowsa}
  Let $\inst I$ be a $k$-PAS over $V$ and $A$.
  If $k\geq \size{A}^{\size{X}}l + |X|$ for  $X\subseteq V$,  then there is $f$ so that $(X,f)$ has $l$-property $P$.
  \end{prop}
  \begin{proof}
    Suppose, for a contradiction, that every $f\in A^X$ has property $\neg P$ and let 
    let $W_f$ be a witness of  $\neg P$ for $f$.
    Choose $W\in\binom{A}{k}$ so that $X\subseteq W$ and moreover $W_f\subseteq W$ for every $f\in A^X$ (which is possible by the assumed inequality).
    Fix an arbitrary $g\in\inst I(W)$ and note the contradiction: 
    since $W_{\proj_X g}\subseteq W$ we derived $\proj_X g\neq \proj_X g$.
  \end{proof}
  \noindent The next proposition concerns PASes of special form and will serve as a base for an inductive proof.
  It says that, given suitable parameters,
  if $\neg Q$ is present throughout the PAS of value $1$, then an $m$-solution can be found.
  \begin{prop}\label{prop:sol}
  Let $\inst I$ be a $k$-PAS over $V$ and $A$ and $\val(\inst I) =1$.
    If for every $v\in V$ there exists $a\in A$ such that $(v,a)$ has $l$-property $\neg Q$, 
    then 
    $\inst I$ is $\lfloor \frac{k}{l+1}\rfloor$-solvable.
  \end{prop}
  \begin{proof}
    Let $s$ be the function mapping $v$ to the associated $a$ with $l$-property $\neg Q$.
    Choose $W\subseteq V$ 
    so that $|W|\leq\frac{k}{l+1}$ and let $W'\in\binom{V}{k}$ include $W$ as well as a witness for every $a\in W$.
    By property $\neg Q$, the projection $\proj_W \inst I(W')$ is equal to $\set{\proj_W s}$ and the proposition is proved.
  \end{proof}

  \subsection{Refining PASes}
    We will be repeatedly performing a construction called \emph{refining a PAS}:
    given 
    \begin{itemize}
      \item a $k$-PAS denoted by $\inst I$, 
      \item a number $l\leq k$, and 
      \item a mapping $\ex: {\binom{V}{l}} \rightarrow {\binom{V}{k}}$ satisfying $U\subseteq \ex(U)$,
    \end{itemize}
    we define an $l$-PAS $\inst J$ by putting $\inst J(U)$ to be $\proj_U(\inst I(\ex(U))$.
    That is, to define the value of $\inst J$ on $U$ we extend it to $\ex(U)$, use $\inst I$ to obtain an associated set of functions, and restrict these functions  to $U$.
    It follows from the definition that every $m$-solution of $\inst J$ is an $m$-solution of $\inst I$.

    The next proposition is very similar to \cref{prop:sol},
    but will be applied if a PAS has value greater than one.
    It states that if property $\neg Q$ can be found ``everywhere'', then the PAS can be turned to a consistent sequence of two PASes.

    \begin{prop}\label{prop:to1}
      Let $\inst I$ be a $k$-PAS.
      If every $X\in\binom{V}{k'}$ has an $f$ with $l$-property $\neg Q$ and $k\geq k''+ \binom{k''}{k'}l$,
      then there exists a $k'$-PAS $\inst I'$ of value  $1$
      and a $k''$-PAS $\inst I''$, a refinement of $\inst I$,
      so that $(\inst I'',\inst I')$ is compatible.
    \end{prop}
    \begin{proof}
      Define $\inst I'$ by putting $\inst I'(X) = \set{f}$ where $f$ has $l$-property $\neg Q$ for $X$ witnessed by $W_X$.
      For each $Y\in\binom{V}{k''}$ put $\ex(Y)$ to be any set of size $k$ including $Y\cup\bigcup_{X\in\binom{Y}{k'}} W_X$.
      Define a refinement of $\inst I$ according to $\ex$ and call it $\inst I''$.
      The definition of property $\neg Q$ provides compatibility of $(\inst I'', \inst I')$.
    \end{proof}

    \subsection{Putting things together, i.e., a proof of \texorpdfstring{\cref{thm:main}}{theMain Theorem}}
      We fix $A$ and $m$ and the sequence of values $(d_0,\dotsc, d_r)$~($r\geq 1$).
      We claim that there exist $k_0,\dotsc, k_r$ such that 
      every sequence of compatible PASes 
      $(\inst I_0,\dotsc,\inst I_r)$ such that 
      $\inst I_i$ is a $k_i$-PAS and $\val(\inst I_i)\leq d_i$ is $m$-solvable.

      The general idea is to transform the sequence 
      into another compatible sequence.
      This is achieved in two steps.
      In the first step we look at every, except for $\inst I_0$, PAS $\inst I_i$ separately.
      If the $\neg Q$ property ``can be found everywhere'' in the PAS, then either
      \cref{prop:sol} provides an $m$-solution, 
      or \cref{prop:to1} offers a reduction
      to a sequence $(\inst I'',\inst I')$ with $\val(\inst I') =1$.
      In the remaining case, we
      refine $(\inst I_0,\dotsc, \inst I_{r})$
      one by one, to obtain a new sequence 
      and then add a twist 
      that makes value of the new PAS at position zero at most $d_0-1$.
      This finishes the reduction.
      Note that the second case cannot happen when $d_0=1$, 
      and that \cref{prop:to1} can be applied at most once during the procedure.

      Formally, we proceed by induction on the sequence of values $(d_0,\dotsc,d_r)$.
      While working on $(d_0,\dotsc,d_r)$ we need the result established for
      \begin{itemize}
        \item sequence of values $(d_i,1)$ for each $i\geq 1$ with $d_i\geq 2$, and
        \item the sequence of values $(d_0-1,d_1,\dotsc, d_r)$, if $d_0\geq 2$.
      \end{itemize}
      In particular, to establish the base of induction, one needs to prove the result for sequence of values $(1,1)$.

      Let us fix a sequence $(d_0,\dotsc,d_r)$ and begin the proof.
      If $d_0\geq 2$ we let $p_0,\dotsc, p_r$ to be the sequence 
      provided by an inductive assumption, i.e.,
      if $(\inst J_0,\dotsc,\inst J_r)$ is a compatible sequence, 
      $\inst J_i$ is a $p_i$-PAS, and $\val(\inst J_i)\leq d_i$ while $\val(\inst J_0)\leq d_0 -1$,
      then some $\inst J_i$ has an $m$-solution.
      If $d_0 = 1$ we put $p_1=\dotsb=p_r = 1$.
      
      Next, we will 
      construct sequence $(k_0,\dotsc,k_r)$ and an auxiliary sentence $(l_0,\dotsc, l_r)$.
      Both sentences are constructed simultaneously from their last elements, $k_r$ and $l_r$, to the first ones.
      The sequences are defined as follows.
      \begin{itemize}
        \item For $i$ equal to $r, r-1, \dotsc, 1$ we put $l_i = p_i + \sum_{j\geq i+1}\binom{p_i}{p_j}(k_j-p_j)$~(if $i=r$ the sum contributes nothing)
          and compute $k_i$ from $l_i$:
          \begin{itemize}
            \item if $d_i=1$ we put $k_i = (l_i+1)m$~(we also fix $k'_i=1$ to be used later),
            \item otherwise we set $k''_i,k'_i$ to be the arities, which work for the sequence of values $(d_i,1)$,
              and put $k_i = k_i'' + \binom{k_i''}{k_i'}l_i$.
          \end{itemize}
        \item Finally, $l_0 = p_0 +\sum_{1\leq j}\binom{p_0}{p_j}(k_j-p_j)$~(i.e., exactly as above)
          and  let $k_0  =  \sum_{1\leq j}k'_j + \size{A}^{\sum_{1\leq j}k'_j}$.
      \end{itemize}

      In the first step of the proof, we assume $i\geq 1$ and work with every $k_i$-PAS $\inst I_i$ separately~(we ignore $\inst I_0$ in this step).
      If $d_i =1$ and
      \cref{prop:sol} can be applied to $\inst I_i$ with the parameter $l_i$ we obtain an $m$-solution and the proof is done.
      From now on we assume this is not the case and thus, if $d_i=1$,
      there exists $v_i$ 
      so that for all $a_i$ the pair $(v_i,a_i)$ has $l_i$-property $Q$ for $\inst I_i$.
      We put $X_i = \set{v_i}$ for later reference.

      If $d_i\geq 2$, the numbers $k_i'',k_i'$ provide an $m$-solution 
      for $(\inst I'',\inst I')$ whenever $\inst I''$ is a $k''$-PAS, $\inst I'$ is a $k'$-PAS, and $\val(\inst I'')=d_i$ while $\val(\inst I')=1$. 
      If \cref{prop:to1} can be applied to $\inst I_i$ with parameters $l_i,k_i'$, and $k_i''$~(in places of $l,k',k''$ respectively),
      we can reduce the problem to the pair of PASes provided by \cref{prop:to1}
      and a solution exists by inductive assumption.
      From now on we assume this is not the case and thus
      there exists $X_i$ of size $k'_i$
      with all $f\in A^{X_i}$ having $l_i$-property $Q$ in $\inst I_i$.

      In the second step, we put $X =\bigcup_i X_i$
      and use \cref{prop:allowsa} to find $f\in A^X$ so that $(X,f)$ has $l_0$-property $P$ in $\inst I_0$.
      Then, for every $i\geq 1$ we put $f_i = \proj_{X_i} f$.

      The last part is direct if a bit technical.
      We will define sequence $(\inst J_0,\dotsc,\inst J_r)$ 
      such that $\inst J_i$ is a $p_i$-PAS and a refinement of $\inst I_i$ for $i\geq 1$~(in particular $\val(\inst J_i)\leq\val(\inst I_i)$).
      The $p_0$-PAS $\inst J_0$ is a refinement of $\inst I_0$ with enough functions removed
      so that $\val(\inst J_0)\leq d_0-1$.
      To fix these refinements,
      we need to define a  map $\ex_i$ for every PAS $\inst I_i$.
      
      We start with $i=r$ and progressively define $\ex_i$ for smaller $i$. 
      For $Y\in \binom{V}{p_i}$ we put $\ex_i(Y)$ to be a set $U$ provided 
      by property $Q$ for $(X_i,f_i)$ and 
      \begin{equation*}
        W = Y \cup \bigcup_{j\geq i+1}\bigcup_{Z\in\binom{Y}{p_j}}\ex_j(Z).
      \end{equation*}
      In the PAS $\inst I_0$ we additionally remove everything that arose from functions extending $f$, that is,
       $\inst J_0(Y) = \proj_Y \inst I_0(\set{g\in \ex_0(Y)| \proj_X g\neq f})$.

      It remains to confirm that the sequence $(\inst J_0,\dotsc,\inst J_r)$ is compatible.
      Let $Y_0\supseteq Y_1\supseteq \dotsb \supseteq Y_r$ be a sequence~
(for the $\inst J_i$, i.e., of sizes $p_0, \dotsc, p_r$,   respectively) and
      consider the sequence $\ex_1(Y_0), \dotsc, \ex_r(Y_r)$~(for the $\inst I_i$, i.e., of sizes $k_0, \dotsc, k_r$,   respectively).
      Note that, by the definition of $W$ in the previous paragraph,
      we have $\ex_0(Y_0)\supseteq \dotsb \supseteq \ex_r(Y_r)$.
      
      By the compatibility of the original sequence, we get
      $i<j$ and $g\in \inst I_j(\ex_j(Y_j)) \cap \proj_{Y_j}\inst I_i(\ex(Y_i))$.
      If $1\leq i<j$, the conclusion is now immediate:
      $\proj_{Y_j} (g)\in \inst J_j(Y_j) \cap \proj_{Y_j}\inst J_i(Y_i)$.
      If $0=i<j$, we additionally need to make sure that that the element 
      $g'\in\inst I_0(\ex_0(Y_0))$
      satisfying $\proj_{Y_j}(g')=g$ satisfies $\proj_X(g')\neq f$.
      This fact follows from the choice of $f_j$ and  $X_j$:
      $X_j\subseteq \ex_j(Y_j)$ and, by  property $Q$,
      $\proj_{X_j} g\neq f_j = \proj_{X_j} f$.
      Clearly $\proj_{X_j}(g') = \proj_{X_j}(g)\neq \proj_{X_j} f$ as required, and the proof is complete (noticing that the procedure in the proof gives a polynomial time algorithm).

      Note that in the case $d_0=1$, we would obtain a compatible sequence containing a PAS with value $0$.
      This is clearly impossible and shows that the second case of the proof cannot happen if $d_r=1$.
      In particular, in the base case of induction~(i.e., with the sequence of values $(1,1)$),
      the only possible scenario is that the application of \cref{prop:sol} in the first step provides an $m$-solution.
 
\section{Appendix: Proof of \texorpdfstring{\cref{thm:drhomo}}{the new reduction}} \label{sec:proofRed}
\subsection{Polymorphisms of general arity} \label{subsec:generalArity}
It is convenient to slightly extend the notion of minions and
polymorphism so that the arity can be any set, not just a natural number. 
This way we can avoid ad hoc (and confusing) choices of bijections between $X$ and $[|X|]$.

An $X$-ary \emph{polymorphism} of a PCSP template $(\relstr A, \relstr B)$, where $X$ is a finite nonempty set, is a map $t: A^X \to B$ such that for any relation $R_i$ of $\relstr{A}$ and any $[\ar i] \times X$ matrix $Z \in A^{[\ar i] \times X}$ whose each column $Z[ \bla, x]$ is in $R_i$, the tuple $t(Z)$ obtained by applying $t$ to the rows ($Z[j, \bla]$) is in the corresponding relation $S_i$ of $\relstr{B}$. Note that an $n$-ary polymorphism as defined in \cref{sec:PCSPs} is the same as an $[n]$-ary polymorphism according to this extended definition.

Let $t:A^X\rightarrow B$ and $\pi:X\rightarrow Y$. We say that $s:A^Y\rightarrow B$ is a \emph{minor}~(or $\pi$-minor, if $\pi$ matters) of $t$ and write $t\minor{\pi} s$ if $s(f) = t(f\pi)$ for every $f \in A^Y$.
For the sake of clartiy we extend the minor notation: instead $t\minor{\pi}s$~(as above) we will sometimes write $t\Mminor{\pi}{X}{Y}s$
to stress the fact that $t$ is $X$-ary, $s$ is $Y$-ary and $\pi$ is viewed as mapping $X$ into $Y$.
Note that polymorphisms of a PCSP template (of general arity) are still closed under taking minors. We also extend the definitions of minion, polymorphism minion, and minion homomorphism in the obvious way to accommodate functions of any arity.\footnote{If one defines $\minion MX$ as the set of all $X$-ary polymorphisms and $\minion M\pi(t) = s$ for $t \minor{\pi} s$ as above, $\minion M$ becomes a functor from the category of nonempty finite sets to itself. Minor homomorphisms then exactly correspond to natural transformations. However, we follow the more standard notation in this paper.}

A simple but crucial property of polymorphisms is that it maps tuples  of (partial) solutions of the strict version of an instance to (partial) solutions of the relaxed instance. More precisely,
if $\Phi$ is an instance of $\PCSP(\relstr{A},\relstr{B})$, 
$t$ is an $X$-ary polymorphism of $(\relstr{A},\relstr{B})$ and $Z$ is a $U \times X$ matrix whose each column $Z[\bla,x]$ is a partial solution (a map $U\to A$) to $\Phi^{\relstr{A}}$, then $t(Z)$ (a map $U \to B$) is 
a partial solution to $\Phi^{\relstr{B}}$.

\subsection{Free PCSP templates}

 Let $\minion M$ be any minion, $C$ any nonempty finite set, and $R\subseteq C^m$ any nonempty relation.
  We follow~\cite{BBKO} and define 
  an $m$-ary relation $\free{C}{\minion M}(R)$
  over the set of $C$-ary functions of $\minion M$:
  denoting $\proj_1, \dots, \proj_m$ the projection maps $R \to C$, we define  
  \begin{equation*}
    (s_1,\dotsc,s_m)\in\free{C}{\minion M}(R)\text{ iff } \exists t\in\minion M\mbox{ (arity $R$) }: \ t \Mminor{\proj_i}{R}{C} s_i \text{ for all $i$}.
  \end{equation*}
(i.e. the arity of $s_i$ does not depend on the set of elements that actually appear on position $i$ in the tuples in $R$).
  
  For a minion $\minion M$, a fixed positive integer $m$, and a (finite nonempty) set $C$ we define the \emph{$m$-ary free PCSP template on $C$}
  by $(\freeL{m}{C},\freeR{m}{C}{\minion M})$, where
  \begin{itemize}
  \item $\freeL{m}{C} = (C;R_1,\dotsc,R_?)$ where the relations list every relation on $A$ of arity at most $m$ and 
  \item $\freeR{m}{C}{\minion M}$ is build on the set of $C$-ary members of $\minion M$ 
    and the relation $S_i$ corresponding to $R_i$ in $\freeL{m}{A}$ is $S_i = \free{C}{\minion M}(R_i)$ in $\freeR{m}{C}{\minion M}$.
  \end{itemize}
  
  The following reduction will serve as the second reduction in the proof of \cref{thm:drhomo}.
  
  \begin{theorem}[\cite{BBKO}] \label{thm:free}
  Let $(\relstr A,\relstr B)$ be a PCSP template, 
$m \in \EN$, and $C\neq\emptyset$ be finite. Then $\PCSP(\freeL{m}{C},\freeR{m}{C}{\Pol(\relstr A,\relstr B)})$ is log-space reducible to $\PCSP(\relstr A,\relstr B)$.
  \end{theorem}
  \begin{proof}[Comments on the proof]
     The mapping $\xi: \Pol(\relstr{A},\relstr{B}) \to \Pol(\freeL{m}{C},\freeR{m}{C}{\Pol(\relstr A,\relstr B)})$ defined for an $X$-ary polymorphism $t$ of $(\relstr{A},\relstr{B})$ by $\xi(t)(f)(g) = t(gf)$ for every $f \in C^X$, $g \in A^C$ is a minion homomorphism (this is the minion homomorphism $\phi$ from Section 4.1 of~\cite{BBKO}), so the claim follows from \cref{thm:oldred}. 
     
     For the interested readers, we mention that the reduction is the standard long code reduction. It works as follows. For each original variable we introduce a cloud of $A^C$ variables (that are meant to provide the long code of the original variable) and for each original constraint involving relation $R$ we introduce a cloud of $A^R$ variables (meant to provide the long code of a member of $R$). We introduce constraints which say that each cloud determines a polymorphism of $(\relstr{A},\relstr{B})$ and finally we merge suitable variables to ensure satisfaction of the original constraints.  
     
     As a final remark, let us mention that a reduction in the opposite direction works as well, provided $m$ and $C$ are sufficiently large~\cite{BBKO}. 
  \end{proof}

In the proof of~\cref{thm:drhomo} we will use a 2-ary free PCSP template and only use relations that are graphs of maps from a subset of $C$ to $C$. 
The following notation and observation will come in handy.
We denote the identity map by $\id$,
independent on its domain or co-domain, and 
if $\pi$ is a function with domain $C_1$ and co-domain $C_2$ satisfying $C_1\cup C_2\subseteq C$, we treat $\pi$ as a subset of $C^2$.

\begin{lemma} \label{lem:partialMap}
Let $\minion{M}$ be a minion, $C$ a (finite nonempty) set, $C_1$ and $C_2$ subsets of $C$, and $\pi: C_1 \to C_2$ a map. 
Then $(s_1,s_2) \in \free{C}{\minion{M}}(\pi)$ if and only if  there exist 
$C_i$-ary members $t_i$ of $\minion{M}$, where $i=1,2$, such that $t_i \Mminor{\id}{C_i}{C} s_i$, and  $t_1 \Mminor{\pi}{C_1}{C_2} t_2$.
\end{lemma}
\begin{proof}
 Straightforward.  
\end{proof}

Note that for any injective $\iota: X \to Y$ and and $Y$-ary $s$ there exists at most one $X$-ary $t$ with $t \minor{\iota} s$. In particular, the $t_i$ in the lemma are unique. 

\subsection{The proof}

\drHomoReduction*

\begin{proof}
   Given an instance $\Phi$ of $\PCSP(\relstr A_2,\relstr B_2)$ we produce $\Psi$, which is an instance of
   $\PCSP(\freeL{2}{C},\freeR{2}{C}{\minion{M}_1})$~
   (for $C$ which is fixed and does not depend on $\Phi$). 
   Then we use the reduction from \cref{thm:free} to produce an instance of $\PCSP(\relstr A_1,\relstr B_1)$.
   
   Let $k_0, \dotsc, k_r$ be the number provided by \cref{cor:lcgt} for $B_2$~(in place of $A$) and let $m$ be the maximal arity of a relation in $\relstr{A}_2$ (or $\relstr{B}_2$).
   The last thing we need to fix is $C$;
   it would be most convenient to have a different domain for each variable of $\Psi$ 
   since then we could define the reduction in essentially the same  way as in the proof of \cref{thm:babyLPCP}. 
   However, we do not have such a freedom (see the remarks in \cref{subsec:rants}) and we set $C$ to be  an arbitrary set of size at least $|A_2^{k_0}|$.

  Our reduction transforms an instance $\Phi$ of $\PCSP(\relstr A_2,\relstr B_2)$ with a set of variables $V$ 
  to an auxiliary instance $\Psi'$ and then to an instance $\Psi$ of $\PCSP(\freeL{2}{C},\freeR{2}{C}{\minion{M}_1})$.
  The set of variables, of both $\Psi'$ and $\Psi$, is 
  $\Vars V = \bigcup_i \Vars V_i$ where $\Vars V_i = \binom{V}{k_i}$.
  For each $U\in\Vars V$ we put $D_U\subseteq A_2^U$ to be the set of partial solutions to $\Phi^{\relstr A_2}$~(the set needs to be non-empty as a solution of $\Phi^{\relstr A_2}$ is promised).
  The constraints of $\Psi'$ are 
  introduced for each pair of elements of $\Vars V$ satisfying $U\supseteq W$; 
  we put $((U,W),\pi_{U,W})$ where $\pi_{U,W} = \set{ (f,g)\in D_U\times D_W| \proj_W f=g})$.
  Note that $\pi_{U,W}$ is in fact a function from $D_U$ into $D_W$~(as a restriction of a partial solution is a partial solution).
  The only problem with the instance $\Psi'$ is that its domain is huge, and the reduction requires a domain of constant~(i.e. independent on $\Phi$) size.
  This problem is resolved in a rather pedestrian fashion.

  For each $U\in\Vars V$ we fix $\sigma_U$ as a bijection between $D_U$ and some $C_U\subseteq C$.
  Then for each constraint $((U,W),\pi_{U,W})$ of $\Psi'$ we introduce into $\Psi$ the constraint
  $((U,W),\sigma(\pi_{U,W}))$ where $\sigma (\pi_{U,W}) = \set{(\sigma_U(f),\sigma_W(g))| (f,g)\in \pi_{U,W}}$.
  In essence, the last transformation renames the elements of the domain without altering the structure of the instance.
  Thus we obtain an instance with domain $C$ and reduction is finished.

  Soundness of the reduction is, again, immediate: if $h: V \to A$ is a solution to $\Phi^{\relstr{A}_2}$ , 
  then $s$ defined by $s'(U) = \proj_U h$ is a solution to $\Psi'$ and $s(U)=\sigma_U(s'(U))$ a solution to $\Psi$.

  For the completeness part, take a solution $s$ of $\Psi^{\freeR{2}{C}{\minion{M}_1}}$, that is, 
  for each $U \in X$, $s(U)$ is a $C$-ary  member of $\minion{M}_1$,  
  and $s$ satisfies all the constraints, that is, $(s(U),s(W)) \in \free{\minion M_1}{C}(\sigma(\pi_{U,W}))$ for any $U\supseteq W$ in $\Vars V$. 
  \cref{lem:partialMap} now delivers two pieces of information:
  \begin{itemize}
    \item For each $U\in\Vars V$ there exist unique $t'_U\in\minion{M}_1$ such that $t'_U \Mminor{\id}{C_U}{C} s(U)$.
    \item For any $U \supseteq W$ in $\Vars V$ we have
      $t'_U \Mminor{\sigma(\pi_{U,W})}{C_U}{C_W} t'_W$. 
  \end{itemize}
  By defining a $D_U$-ary $t_U$ by $t'_U \Mminor{\sigma_U^{-1}}{C_U}{D_U} t_U$ we finally obtain 
  \begin{equation}\label{eq:xx}
    t_U \Mminor{\pi_{U,W}}{D_U}{D_W} t_W \text{ for any $W,U\in \Vars V$ such that $U \supseteq W$.}
  \end{equation}

 It remains to decode the $t_U$ into a sequence of PASes. 
  To this end we first define  (for any $U \in \Vars V$) a $U \times D_U$ matrix $Z_U$ 
  by $Z_U[u,f] = f(u)$ for $u \in U$ and $f \in D_U$. 
  Observe that each column $Z_U[\bla,f]$ of this matrix is a partial solution of $\Phi^{\relstr A_2}$, namely, $f$. 
 
 For $0 \leq i \leq r$ and $U \in \Vars V_i$ define $\inst{I}_i(U) = \set{q(Z_U)| q \in \xi(t_U)}$, and note that every element of $\inst{I}_i(U)$ is a partial solution to $\Phi^{\relstr{B}_2}$ (recall the remark in the final paragraph of \cref{subsec:generalArity}), that the size of $\inst{I}_i(U)$ is at most $d$ (by the definition of $(d,r)$-homomorphism), and that $\inst{I}_i$ is a $k_i$-PAS. 
 
 It remains to verify consistency. Let $U_0 \supseteq U_1 \supseteq \cdots \supseteq U_r$ be subsets of $V$ (of sizes $k_0, k_1, \dots, k_r$) and consider the chain of minors
  $t_{U_0}\minor{\pi_{U_0,U_1}}t_{U_1}\minor{\pi_{U_1,U_2}}\dotsb\minor{\pi_{U_{r-1},U_r}}t_{U_r}$ that we have from \cref{eq:xx}. By the definition of $(d,r)$-homomorphism, there 
  exist $i<j$ and $q_i \in \xi(t_{U_i})$, $q_j \in \xi(t_{U_j})$ such that $q_i \minor{\pi_{U_i,U_j}} q_j$.
  
  We claim that $\proj_{U_j} q_i(Z_{U_i}) = q_j(Z_{U_j})$ -- then this element witnesses $\inst{I}(U_j) \cap \proj_{U_j} \inst{I}(U_i)$ and consistency is established. 
  To prove the claim, we need to verify that, for each $u \in U_j$, $q_i$ applied to the $u$-th row of $Z_{U_i}$~(i.e. a map $w_i : D_{U_i}\rightarrow A_2$ mapping $f\mapsto f(u)$)
  gives the same element of $B_2$ as $q_j$ applied to the $u$-th row (denoted $w_j$) of the matrix $Z_{U_j}$.

  Since $q_i \minor{\pi_{U_i,U_j}} q_j$, 
  we have $q_j(w_j) = q_i(w_j \pi_{U_i,U_j})$ by definition of minors.
  It is enough to verify $q_i(w_i) = q_i(w_j \pi_{U_i,U_j})$. 
  But this is clear -- for any $f \in D_{U_i}$ we have $w_i(f) = f(u)$ (by definition of the matrix) 
  and $w_j\pi_{U_i,U_j}(f) = (\pi_{U_i,U_j}(f))(u) = (\proj_{U_j} f) (u) = f(u)$.

  We have shown that the value $\val_{k_1,\dotsc,k_r}(\Phi^{\relstr{B}_2})$ is at most $d$,
  and by \cref{cor:lcgt} it must be $1$, which makes $\Phi^{\relstr{B}_2}$ solvable.
  This finishes the proof of soundness and of \cref{thm:drhomo}
  \end{proof}

\subsection{Multisorted PCSP} \label{subsec:rants}

There are two phenomena apparent from the proof (among other contexts) worth a short note. The first one is that it would be convenient to allow multiple domains for variables, e.g., to work with multi-sorted relational structures. The second one is that we have only used relations that are graphs of functions (the functions were partial, but they would become proper had we multiple sorts). If we do these modification to the definition of a CSP template (i.e., allow multiple sorts but only binary constraints that are graphs of functions\footnote{\dots so a template for CSP can be defined as a finite subcategory of the category of finite sets -- this is perhaps a nicer definition than via structures}), the framework we get would become richer: we could still express all the (P)CSPs (by replacing relations by projection maps) and, moreover, Layered Label Cover would become a CSP, the gap version from \cref{thm:babyLPCP} with $r=1$ would become a PCSP, and the gap version from \cref{thm:LPCP} with $r=1$ would become a Valued PCSP. Note, however, that the gap versions with $r>1$ would still not be (V)PCSPs. Is this because Gap Layered Label Cover is an unnatural problem which will eventually become obsolete, or is it hinting us toward a better framework?

\bibliographystyle{plainurl}

\end{document}